\documentclass[journal,twocolumn,a4paper]{IEEEtran}

\usepackage{amsfonts,amssymb} 
\usepackage{amsmath}
\usepackage{graphicx}
\usepackage{comment}
\usepackage{url}
\usepackage{xcolor}

\begin{document}

\newtheorem{theorem}{Theorem}
\newtheorem{corollary}{Corollary}
\newtheorem{proposition}{Proposition}
\newtheorem{lemma}[theorem]{Lemma}
\newtheorem{assumption}{Assumption}
\newtheorem{definition}{Definition}
\newtheorem{example}{Example}

\newenvironment{proof}[1][\scshape{Proof:}]
{\begin{trivlist} \item[\hskip \labelsep {\small #1}]}
{\begin{flushright}$\blacksquare$\end{flushright}\end{trivlist}}

\def\vW{ {\vec{W}} }
\def\vw{ {\vec{w}} }
\def\vy{ {\vec{y}} }
\def\vx{ {\vec{x}} }

\def\DD{ D^* }
\def\NS{ {N_S} }
\def\GS{ {G_S} }
\def\GSprime{ {G_{S'}} }
\def\Umin{ {\text{U-min} } }
\def\stop{{ T(U,V,\vw)} }

\def\Gopt{ {G_{\text{opt}} } }

\def\Rv{ {R^{(v)}} }
\def\Wv{ W^{(v)} }
\def\Gv{ {G^{(v)}} }
\def\Gvi{ {G^{(v_i)}} }
\def\Gvone{ {G^{(v_1)}} }

\def\deltav{ {\delta^{(v)}} }
\def\gammav{ {\gamma^{(v)}} }
\def\gammavi{ {\gamma^{(v_i)}} }

\def\ne{ {n_e} }

\def\HG{ {H_G} }

\def\LambdaG{ {\Lambda_G} }
\def\LambdaGv{ {\Lambda_G^{(v)}} }
\def\LambdaGvone{ {\Lambda_G^{(v_1)}} }
\def\LambdaGopt{ \Lambda_{G_\text{opt}} } 
\def\IGopt{ I_{G_\text{opt}} } 

\def\A{ {\mathbb{A} }}
\def\N{ {\mathbb{N} }}
\def\R{ {\mathbb{R} }}

\def\TURN{ { \gamma } }
\def\E{{E}}
\def\P{{P}}
\def\T{{T_{\epsilon,k}}}
\def\T{{T^\epsilon_k}}
\def\U{U^\epsilon}
\def\W{W^\epsilon}
\def\p1{ {g} }
\def\lplus{l^+}
\def\lminus{l^-}
\def\norm{ \eta }

\def\SIZE{ {N_k(l) } }
\def\SIZEp{ {N_k(\lplus) } }
\def\SIZEm{ {N_k(\lminus) } }
\def\SIZEt{ {N_k(l^*) } }
\def\SIZEemax{ {N_k(l^{\epsilon, *, (k)}} }
\def\SIZEmax{ {N_k(l^{*, (k)}) } }

\def\Na{ {n_k(w,a)} }
\def\La{ {L_k} }
\def\Le{ {L_{\epsilon, k}} }
\def\Les{ {L_{\epsilon}} }

\sloppy

\title{Multi-user guesswork and brute force security}

\author{
    \IEEEauthorblockN{Mark M. Christiansen and Ken R. Duffy}
   \IEEEauthorblockA{Hamilton Institute\\
     National University of Ireland Maynooth\\
     Email: \{mark.christiansen, ken.duffy\}@nuim.ie}\\
  \and
\IEEEauthorblockN{Fl\'avio du Pin Calmon and Muriel M\'edard }
  \IEEEauthorblockA{Research Laboratory of Electronics\\
Massachusetts Institute of Technology\\
     Email: \{flavio, medard\}@mit.edu}
  
}

\maketitle

\begin{abstract}

The Guesswork problem was originally motivated by a desire to
quantify computational security for single user systems. Leveraging
recent results from its analysis, we extend the remit and utility
of the framework to the quantification of the computational security
of multi-user systems. In particular, assume that $V$ users
independently select strings stochastically from a finite, but
potentially large, list. An inquisitor who does not know which
strings have been selected wishes to identify $U$ of them. The
inquisitor knows the selection probabilities of each user and is
equipped with a method that enables the testing of each (user,
string) pair, one at a time, for whether that string had been
selected by that user.

Here we establish that, unless $U=V$, there is no general strategy
that minimizes the distribution of the number of guesses, but in
the asymptote as the strings become long we prove the following:
by construction, there is an asymptotically optimal class of
strategies; the number of guesses required in an asymptotically
optimal strategy satisfies a Large Deviation Principle with a rate
function, which is not necessarily convex, that can be determined
from the rate functions of optimally guessing individual users'
strings; if all users' selection statistics are identical, the
exponential growth rate of the average guesswork as the string-length
increases is determined by the specific R\'enyi entropy of the
string-source with parameter $(V-U+1)/(V-U+2)$, generalizing the known
$V=U=1$ case; and that the Shannon entropy of the source is a lower
bound on the average guesswork growth rate for all $U$ and $V$,
thus providing a bound on computational security for multi-user
systems. Examples are presented to illustrate these results
and their ramifications for systems design.

\end{abstract}

\section{Introduction}
\let\thefootnote\relax\footnote{
F.d.P.C. sponsored by the Department of Defense under Air
Force Contract FA8721-05-C-0002. Opinions, interpretations,
recommendations, and conclusions are those of the authors and are
not necessarily endorsed by the United States Government. Specifically,
this work was supported by Information Systems of ASD(R\&E). M.M.
and K.D. were supported in part by a Netapp faculty fellowship.}

The security of systems is often predicated on a user or application
selecting an object, a password or key, from a large list. If an
inquisitor who wishes to identify the object in order to gain access
to a system can only query each possibility, one at a time, then
the number of guesses they must make in order to identify the
selected object is likely to be large. If the object is selected
uniformly at random using, for example, a cryptographically secure
pseudo-random number generator, then the analysis of the distribution
of the number of guesses that the inquisitor must make is trivial.

Since the earliest days of code-breaking, deviations from perfect
uniformity have been exploited. For example, it has long since been
known that human-user selected passwords are highly non-uniformly
selected, e.g. \cite{malone12}, and this forms the basis of dictionary
attacks. In information theoretic security, uniformity of the string
source is typically assumed on the basis that the source has been
compressed. Recent work has cast some doubt on the appropriateness
of that assumption by establishing that fewer queries are required
to identify strings chosen from a typical set than one would
expect by a na\"ive application of the asymptotic equipartition
property. This arises by exploitation of the mild non-uniformity
of the distribution of strings conditioned to be in the typical set
\cite{Christiansen13a}.

If the string has not been selected perfectly uniformly, but with
a distribution that is known to the inquisitor, then the quantification
of security is relatively involved. Assume that a string, $W_1$,
is selected stochastically from a finite list, $\A = \{0,\ldots,m-1\}$.
An inquisitor who knows the selection probabilities, $\P(W_1=w)$
for all $w\in\A$, is equipped with a method to test one string
at a time and develops a strategy,
$G:\A\mapsto\{1,\ldots,m\}$,
that defines the order in which strings are guessed. As the string
is stochastically selected, the number of queries,
$G(W_1)$,
that must be made before it is identified correctly is also
a random variable, dubbed guesswork. Analysis of the distribution
of guesswork serves as a natural a measure of computational security
in brute force determination.

In a brief paper in 1994, Massey \cite{Massey94} established that
if the inquisitor orders his guesses from most likely to least
likely, then the Shannon entropy of the random variable $W_1$ bears
little relation to
the expected guesswork
$E(G(W_1))= \sum_{w\in\A} G(w) P(W_1=w)$,
the average number of guesses required to identify 
$W_1$.
Arikan \cite{Arikan96} established
that if a string, $W_k$, is chosen from $\A^k$ with i.i.d. characters,
again guessing strings from most likely to least likely, then the
moments of the guesswork distribution grow exponentially in $k$
with a rate identified in terms of the R\'enyi entropy of the
characters, 
\begin{align*}
\lim_{k\to\infty} \frac 1k \log E(G(W_k)^\alpha)
	&= (1+\alpha) \log \sum_{w\in\A} P(W_1=w)^{1/(1+\alpha)}\\
	&= \alpha R\left(\frac{1}{1+\alpha}\right) 
	\text{ for } \alpha>0,
\end{align*}
where $R((1+\alpha)^{-1})$ is the R\'enyi entropy of $W_1$ with
parameter $(1+\alpha)^{-1}$.
In particular, the average guesswork grows as the R\'enyi entropy
with parameter $1/2$, a value that is 
lower bounded by Shannon entropy.

Arikan's result was subsequently extended significantly beyond
i.i.d. sources \cite{Malone042,Pfister04,Hanawal11}, establishing
its robustness. In the generalized setting, specific R\'enyi entropy,
the R\'enyi entropy per character, plays the r\^ole of R\'enyi
entropy. In turn, these results have been leveraged to prove that
the guesswork process $\{k^{-1}\log G(W_k)\}$ 
satisfies a Large Deviation Principle (LDP),
e.g. \cite{Lewis95,Dembo98},
in broad
generality \cite{Christiansen13}.
That is, there exists a lower semi-continuous function
$I:[0,\log(m)]\mapsto[0,\infty]$ such that for all
Borel sets $B$ contained in $[0,\log(m)]$
\begin{align}
-\inf_{x\in B^\circ} I(x) &\leq
        \liminf_{n\to\infty} \frac 1k \log 
	P\left(\frac 1k \log G(W_k) \in B\right) \nonumber\\
&\leq
\limsup_{n\to\infty} \frac 1k \log 
	P\left(\frac 1k \log G(W_k) \in B\right) \nonumber\\
        &\leq -\inf_{x\in \bar{B}} I(x),
	\label{eq:LDP}
\end{align}
where $B^\circ$ denotes the interior of $B$ and $\bar{B}$ denotes
its closure. Roughly speaking, this implies $dP(k^{-1} \log G(W_k)
\approx x)\approx \exp(-k I(x)) dx$ for large $k$. In \cite{Christiansen13}
this LDP, in turn, was used to provide direct estimates on the
guesswork probability mass function, $P(G(W_k)=n)$ for
$n\in\{1,\ldots,m^k\}$.  These deductions, along with others described
in Section \ref{sec:brief}, have developed a quantitative framework
for the process of brute force guessing a single string.

In the present work we address a natural extension in this investigation
of brute force searching: the quantification for multi-user systems.
We are motivated by both classical systems, such as the brute force
entry to a multi-user computer where the inquisitor need only
compromise a single account, as well as modern distributed storage
services where coded data is kept at distinct sites in a way where,
owing to coding redundancy, several, but not all, servers need to be
compromised to access the content \cite{oliveira12,Calmon12}.

\section{Summary of contribution}

Assume that $V$ users select strings independently from $\A^k$. An
inquisitor knows the probabilities with which each user selects
their string, is able to query the correctness of 
each 
(user, string)
pair, and wishes to identify any 
subset of size
$U$ of the $V$ strings. The first
question that must be addressed is what is the optimal strategy, the
ordering in which (user, string) pairs are guessed, for
the inquisitor. For the single user system, since the earliest
investigations \cite{Massey94,Arikan96,Merhav99,Pliam00} it has
been clear that the strategy of ordering guesses from 
the
most to least likely string, breaking ties arbitrarily, is optimal
in any reasonable sense. Here we shall give optimality a specific
meaning: that the distribution of the number of guesses required
to identify the unknown object is stochastically dominated by all
other strategies. Amongst other results, for the multi-user guesswork
problem we establish the following:

\begin{itemize}
\item If $U<V$, the existence of optimal guessing strategies,
those that are stochastically dominated by all other strategies, is
no longer assured.
\item By construction, there exist asymptotically optimal strategies
as the strings become long.
\item For asymptotically optimal strategies, we prove a large
deviation principle for their guesswork. The resulting large deviations 
rate function is, in general, not convex and so this result could
not have been established by determining how the moment generating
function of the multi-user guesswork distribution scales in
string-length.
\item
The non-convexity of the rate function shows that, if users' string
statistics are distinct, there may be no fixed ordering of weakness
amongst users. That is, depending on how many guesses are made
before the $U$ users' strings are identified, the collection of
users whose strings have been identified are likely to be
distinct.
\item If all $V$ strings are chosen with the same statistics, then
the rate function is convex and the exponential growth rate of the
average guesswork as string-length increases is the specific R\'enyi
entropy of the string source with parameter
\begin{align*}
\frac{V-U+1}{V-U+2} 
\in\left\{\frac 12,\frac 23,\frac 34,\frac4 5,\frac 56,\ldots\right\}.
\end{align*}
\item For homogeneous users, from an inquisitor's point of view,
there is a law of diminishing returns for the expected guesswork
growth rate in excess number of users ($V-U$).
\item For homogeneous users, from a designer's point of view, coming full
circle to Massey's original observation that Shannon entropy has
little quantitative relationship to how hard it is to guess a single
string, the specific Shannon entropy of the source is a lower bound
on the average guesswork growth rate for all $V$ and $U$.
\end{itemize}

These results generalize both the original guesswork studies, where
$U=V=1$, as well as some of the results in \cite{Merhav99,Hanawal11a}
where, as a wiretap model, the case $U=1$ and $V=2$ with one of the
strings selected uniformly, is considered and scaling properties
of the guesswork moments are established. Interestingly, we shall
show that that setting is one where the LDP rate function is typically
non-convex, so while results regarding the asymptotic behavior of
the guesswork moments can be deduced from the LDP, the reverse is
not true. To circumvent the lack of convexity, we prove the main
result using the contraction principle, Theorem 4.2.1 \cite{Dembo98},
and the LDP established in \cite{Christiansen13}, which itself
relies on earlier results of work referenced above.

\begin{figure}
\includegraphics[scale=0.46]{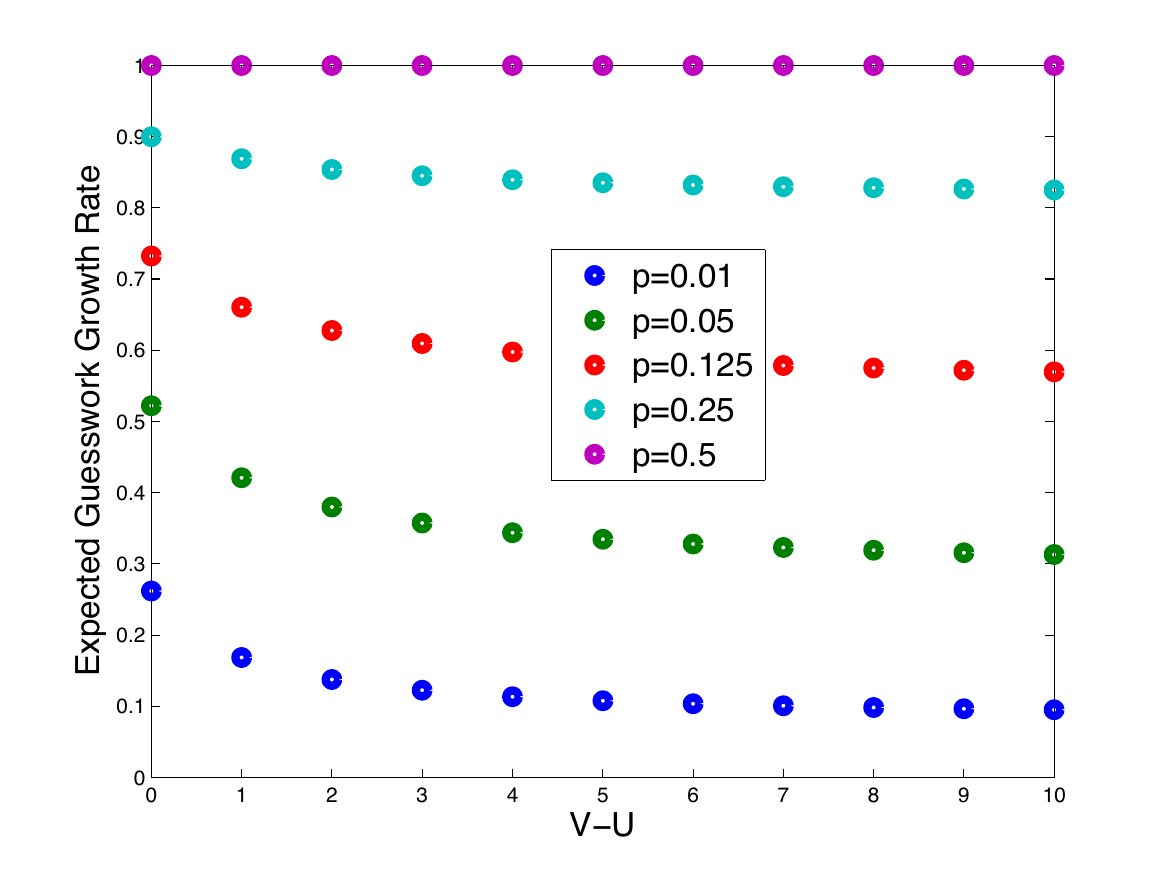}
\includegraphics[scale=0.46]{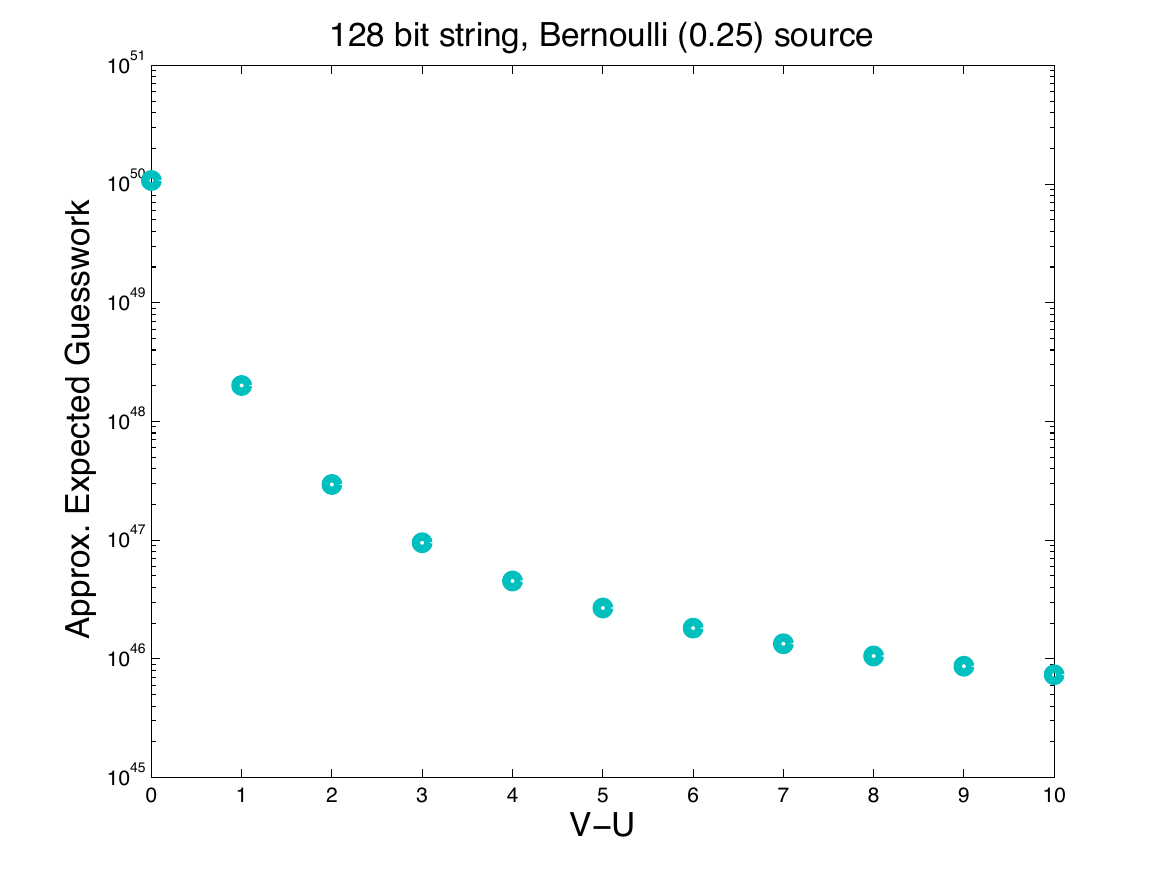}
\caption{Strings created from i.i.d. letters are selected from a
binary alphabet with probability $p$ for one character. Given an
inquisitor wishes to identify $U$ of $V$ strings, the left panel
shows the average exponential guesswork growth rate as a function
of $V-U$, the excess number of guessable strings; the right panel
shows the theoretically predicted approximate average guesswork for
$168$ bit strings, as used in triple DES, as a function of $V-U$,
the excess number of guessable strings.
}
\label{fig:bernoulli_guess}
\end{figure}

\section{
The impact of the number of users on expected guesswork growth rate,
an example
}
\label{sec:firstexample}

As an exemplar that illustrates the reduction in security that comes
from having multiple users, the left panel in Figure
\ref{fig:bernoulli_guess} the average guesswork growth rate for an
asymptotically optimal strategy is plotted for the simplest case,
a binary alphabet with $V$ i.i.d.  Bernoulli string sources. In
order to be satisfied, the inquisitor wishes to identify $U\leq V$
of the strings. The x-axis shows the excess number of guessable
strings, $V-U$, and the y-axis is the $\log_2$ growth rate of the
expected guesswork in string length.  If the source is perfectly
uniform (i.e. characters are chosen with a Bernoulli $1/2$ process),
then the average guesswork growth rate is maximal and unchanging
in $V-U$. If the source is not perfectly uniform, then the growth
rate decreases as the number of excess guessable strings $V-U$
increases, with a lower bound of the source's Shannon entropy.

For a string of length $168$ bits, as used in the triple DES cipher,
and a Bernoulli $(0.25)$ source, the right panel in Figure
\ref{fig:bernoulli_guess} displays the impact that the change in
this exponent has, approximately, on the average number of guesses
required to determine $U$ strings. More refined results for a
broader class of processes can be found in later sections, including
an estimate on the guesswork distribution.

The rest of this paper is organized as follows. In Section
\ref{sec:brief}, we begin with a brief overview of results on
guesswork that we have not touched on so far. Questions of optimal
strategy are considered in Section \ref{sec:strategy}. Asymptotically
optimal strategies are established to exist in Section \ref{sec:asymptote}
and results for these strategies appear in Section \ref{sec:results}.
In Section \ref{sec:mismatch} we present examples where strings
sources have distinct statistics. In Section \ref{sec:ident} we
return to the setting where string sources have identical statistics.
Concluding remarks appear in Section \ref{sec:conc}.

\section{A brief overview of guesswork}
\label{sec:brief}
Since Arikan's introduction of the long string length asymptotic,
several generalizations of its fundamental assumptions have been
explored. Arikan and Boztas \cite{Arikan02} investigate the setting
where the truthfulness in response to a query is not certain. Arikan
and Merhav \cite{Arikan98} loosen the assumption that inquisitor
needs to determine the string exactly, assuming instead that they
only need to identify it within a given distance. That the inquisitor
knows the distribution of words exactly is relaxed by Sundaresan
\cite{Sundaresan07b}, \cite{Sundaresan06} 
and by the authors of \cite{Beirami15}.

Motivated by a wiretap application, the problem of multiple users
was first investigated by Merhav and Arikan \cite{Merhav99} in the
$V=2$ and $U=1$ setting, assuming one of the users selects 
their string uniformly on a reduced alphabet. In \cite{Hayashi06}
Hayashi and Yamamoto extend the results in \cite{Merhav99} to the
case if there is an additional i.i.d. source correlated to the
first, used for coding purposes, while Harountunian and Ghazaryan
\cite{Haroutunian00} extend the results in \cite{Merhav99} to the
setting of \cite{Arikan98}.  Harountunian and Margaryan
\cite{Haroutunian12} expand on \cite{Merhav99} by adding noise to
the original string, altering the distribution of letters. Hanawal
and Sundaresan \cite{Hanawal11a} extend the bounds in \cite{Merhav99}
to a pre-limit and to more general sources, showing that they are
tight for Markovian and unifilar sources.

Sundaresan \cite{Sundaresan07a} uses length functions to identify
the link between guesswork and compression. This result is extended
by Hanawal and Sundaresan \cite{Hanawal11b} to relate guesswork to
the compression of a source over a countably infinite alphabet. In
\cite{Christiansen13a} the authors prove that, if the string is
conditioned on being an element of a typical set the expected
guesswork, is growing 
more slowly
than a simple uniform approximation
would suggest. In \cite{Christiansen13b} the authors consider the
impact of guessing over a noisy erasure channel showing that the
mean noise on the channel is not the significant moment in determining
the expected guesswork, but instead one determined by its R\'enyi
entropy with parameter $1/2$.
Finally, we mention that recent work by Bunte and Lapidoth
\cite{Bunte14} identifies a distinct operational meaning for R\'enyi
entropy in defining a rate region for a scheduling problem.

\section{Optimal strategies}
\label{sec:strategy}
In order to introduce the key concepts used to determine the optimal
multi-user guesswork strategy, we first reconsider the optimal guesswork
strategy in the single user case, i.e.
$U=V=1$.
Recall that $\A=\{0,\ldots,m-1\}$ is a finite set.
  
\begin{definition}
A single user strategy, $S:\A^k\mapsto\{1,\ldots,m^k\}$, is a
one-to-one map that determines the order in which guesses are made.
That is, for a given strategy $S$ and a given string $w\in\A^k$,
$S(w)$ is the number of guesses made until $w$ is queried.
\end{definition}

Let $W_k$ be a random variable taking values in $\A^k$. Assume that
its probability mass function, $P(W_k=w)$ for all $w\in\A^k$, is
known. Since the first results on the topic it has been clear that
the best strategy, which we denote $G$, is to guess from most likely
to least likely, breaking ties arbitrarily. In particular, $G$ is
defined by $G(w)<G(w')$ if $P(W_k=w)>P(W_k=w')$. We
begin by assigning optimality a precise meaning in terms of stochastic
dominance \cite{Lehmann55,denuit06}.

\begin{definition}
A strategy $S$ is optimal for $W_k$ if the random variable $S(W_k)$
is stochastically dominated by $S'(W_k)$, for all strategies $S'$.
That is, if $P(S(W_k)\leq n) \geq P(S'(W_k)\leq n)$ for all strategies
$S'$ and all $n\in\{1,\ldots,m^k\}$.  
\end{definition}

This definition captures the stochastic aspect of guessing by stating
than an optimal strategy is one where the identification stopping
time is probabilistically smallest. One consequence of this definition
that explains its appropriateness is that for any monotone function
$\phi:\{1,\ldots,m^k\}\to\R$, it is the case that $E(\phi(S(W_k)))\leq
E(\phi(S'(W_k)))$ for an optimal $S$ and any other $S'$ (e.g.
Proposition 3.3.17, \cite{denuit06}). Thus $S(W_k)$ has the least
moments over all guessing strategies.
That guessing from most- to least-likely in the single user case is
optimal is readily established.
\begin{lemma}
\label{lem:1opt}
If $V=U=1$, the optimal strategies are those that guess from most
likely to least likely, breaking ties arbitrarily.
\end{lemma} 
\begin{proof}
Consider the strategy $G$ defined above and any other strategy $S$.
By construction, for any $n\in\{1,\ldots,m^k\}$
\begin{align*}
P(G(W_k)\leq n) &= \sum_{i=1}^n P(G(W_k)=i) \\
	& = \max_{w_1,\ldots,w_n} \left(\sum_{i=1}^n P(W_k=w_i)\right)\\
	&\geq \sum_{i=1}^n P(S(W_k)=i) = P(S(W_k)\leq n).
\end{align*}
\end{proof}

In the multi-user case, where (user, string) pairs are queried, a
strategy is defined by the following.
\begin{definition}
A multi-user strategy is a one-to-one map
$S:\{1,\ldots,V\}\times\A^k\mapsto\{1,\ldots,Vm^k\}$ that orders
the guesses of (user, string) pairs.
\end{definition}

The expression for the number of guesses required to identify $U$
strings is a little involved as we must take into account that we
stop making queries about a user once 
their string has been identified.
For a given strategy $S$, let
$\NS:\{1,\ldots,V\}\times\{1,\ldots,Vm^k\}\mapsto\{1,\ldots,m^k\}$ be
defined by
\begin{align*}
\NS(v,n) = |\{w\in\A^k:S(v,w)\leq n\}|,
\end{align*}
which computes the number of queries in the strategy up to $n$
that correspond to user $v$.

The number of queries that need to be made if $U$ strings are
to be identified is
\begin{align*}
\stop = \Umin\left(S(1,w^{(1)}),\ldots,S(V,w^{(V)})\right),
\end{align*}
where $\Umin:\R^V\to\R$ and $\Umin(\vx)$ gives the $U^{\rm th}$
smallest component of $\vx$. 
The 
number of guesses required to identify $U$ components of $\vw
=(w^{(1)},\ldots,w^{(V)})$ is then
\begin{align}
\label{eq:GS}
\GS(U,V,\vw)=
\sum_{v=1}^V \NS\left(v, 
	\min\left(S(v,w^{(v)}),\stop\right) \right).
\end{align}
This apparently unwieldy object counts the number of queries made
to each user, curtailed either when 
their
string is identified or
when $U$ strings of other users are identified.

If $U=V$, equation \eqref{eq:GS} simplifies significantly, as 
$S(v,w^{(v)})\leq\stop$ for all
$v\in\{1,\ldots,V\}$, becoming 
\begin{align}
\label{eq:U=V}
\GS(V,V,\vw)=
\sum_{v=1}^V \NS\left(v, S(v,w^{(v)}) \right),
\end{align}
the sum of the number of queries required to identify each individual
word. In this case, we have the analogous result to Lemma \ref{lem:1opt},
which is again readily established.
\begin{lemma}
\label{lem:UVopt}
If $V=U$, the optimal strategies are those that employ individual
optimal strategies, but with users selected in any order.
\end{lemma} 
\begin{proof}
For any multi-user strategy $S$, equation \eqref{eq:U=V} holds.
Consider an element in the sum on the right hand side,
$\NS\left(v, S(v,w^{(v)}) \right)$.
It can be recognized to be the number of queries made to user $v$
until their string is identified. By Lemma \ref{lem:1opt}, for each
user $v$, for any $S$ this stochastically dominates the equivalent
single user optimal strategy. Thus the multi-user optimal strategies
in this case are the sum of individual user optimal strategies,
with users queried in any arbitrary order.
\end{proof}

The formula \eqref{eq:GS} will be largely side-stepped when we
consider asymptotically optimal strategies, but is needed to establish
that there is, in general, no stochastically dominant strategy if
$V>U$. With $\vW_k=(W_k^{(1)},\ldots,W_k^{(V)})$ being a random vector
taking values in $\A^{kV}$ with independent, not necessarily identically
distributed, components, we are not guaranteed the existence of an
$S$ such that $P(\GS(U,V,\vW_k)\leq n) \geq (\GSprime(U,V,\vW_k)\leq n)$
for all alternate strategies $S'$.

\begin{lemma}
If $V>U$, a stochastically dominant strategy does not necessarily exist.
\end{lemma}
\begin{proof}
A counter-example suffices and so let $k=1$, $V=2$, $U=1$ and
$\A=\{0,1,2\}$. Let the distributions of $W_1^{(1)}$ and $W_1^{(2)}$
be
\begin{center}
\begin{tabular}{|c|c|}
\hline
User 1 & User 2\\
\hline
$P(W_1^{(1)}=0)=0.6$ & $P(W_1^{(2)}=0)=0.5$  \\
$P(W_1^{(1)}=1)=0.25$ & $P(W_1^{(2)}=1)=0.4$ \\
$P(W_1^{(1)}=2)=0.15$ & $P(W_1^{(2)}=2)=0.1$\\
\hline
\end{tabular}
\end{center}
If a stochastically dominant strategy exists, its first guess must
be user $1$, string $0$, i.e. $S(1,0)=1$, so that $P(\GS(1,\vW_1)=1) =
0.6$. Given this first guess, to maximize $P(\GS(1,\vW_1)\leq2)$, the
second guess must be user $1$, string $1$, $S(1,1)=2$, so that 
$P(\GS(1,\vW_1)\leq 2) = 0.85$.

An alternate strategy with $S(2,0)=1$ and $S(2,1)=2$, however,
gives $P(\GSprime(1,\vW_1)=1)=0.5$ and $P(\GSprime(1,\vW_1)\leq 2)=0.9$.
While $P(\GS(1,\vW_1)=1)>P(\GSprime(1,\vW_1)=1)$,
$P(\GS(1,\vW_1)\leq2)<P(\GSprime(1,\vW_1)\leq2)$ and so there is
no strategy stochastically dominated by all others in this case.
\end{proof}

Despite this lack of universal optimal strategy, we shall show that
there is a sequence of random variables that are stochastically
dominated by the guesswork of all strategies and, moreover, there
exists a strategy with identical performance in Arikan's long string
length asymptotic.

\begin{definition}
A strategy $S$ is asymptotically optimal if $\{k^{-1}\log
\GS(U,V,\vW_k)\}$ satisfies a LDP
with the same rate function as a sequence
$\{k^{-1}\log \Upsilon(U,V,\vW_k)\}$ where $\Upsilon(U,V,\vW_k)$
is stochastically dominated by $\GSprime(U,V,\vW_k)$ for all
strategies $S'$.
\end{definition}

Note that $\Upsilon(U,V,\cdot)$ need not correspond to the guesswork
of a strategy.

\section{An asymptotically optimal strategy}
\label{sec:asymptote}

Let $\{\vW_k\}$ be a sequence of random strings, with $\vW_k$ taking values in
$\A^{kV}$, with independent components, $W_k^{(v)}$, corresponding
to strings selected by users $1$ through $V$, although each user's
string 
may 
not be constructed from i.i.d. letters. 
For each individual
user, $v\in\{1,\ldots,V\}$, let $\Gv$ denote its single-user optimal
guessing strategy; that is, guessing from most likely to least likely.

We shall
show that the following random variable, constructed
using the $\Gv$, is stochastically
dominated by the guesswork distribution of all strategies:
\begin{align}
\Gopt(U,V,\vW_k) =
\Umin \left(G^{(1)}(W^{(1)}_k),\ldots,G^{(V)}(W^{(V)}_k)\right).
\label{eq:LB}
\end{align}
This can be thought of as allowing the inquisitor to query, for
each $n$ in turn, the $n^{\rm th}$ most likely string for all users
while only accounting for a single guess and so it does not
correspond to an allowable strategy.

\begin{lemma}
For any strategy $S$ and any $U\in\{1,\ldots,V\}$,
$\Gopt(U,V,\vW_k)$ is stochastically dominated by $\GS(U,V,\vW_k)$.
That is, for any any $U\in\{1,\ldots,V\}$ 
and any 
$n\in\{1,\ldots,m^k\}$
\begin{align*}
P(\Gopt(U,V,\vW_k)\leq n) \geq P(\GS(U,V,\vW_k)\leq n).
\end{align*}
\end{lemma}
\begin{proof}
Using equation \eqref{eq:GS} and the positivity of its
summands, for any strategy $S$
\begin{align*}
&G_S(U, V, \vw)\\
&\ge\Umin (\NS(1, S(1, w^{(1)})),\ldots, \NS(V, S(V, w^{(V)}))).  
\end{align*}
As for each $v\in\{1,\ldots,V\}$, $G^{(v)}(W_k^{(v)})$ is stochastically
dominated by all other strategies,
\begin{align*}
P(G^{(v)}(W_k^{(v)})\le n)\ge P(\NS(v, S(1, W_k^{(v)}))\le n).
\end{align*}
Using equation \eqref{eq:LB}, this implies that
\begin{align*}
&P(\Gopt(\vW_k)\le n)\\
&\ge P(\Umin (\NS(1, S(1, W_k^{(1)})),\ldots, \NS(V, S(V, W_k^{(V)})))\le n)\\
&\ge P(G_S(U, V, \vW_k)\le n),
\end{align*}
as required.
\end{proof}

The strategy that we construct that will asymptotically meet the
performance of the lower bound is to round-robin the single user
optimal strategies. That is, to query the most likely string of one
user followed by the most likely string of a second user and so
forth, for each user in a round-robin fashion, before moving to the
second most likely string of each user. An upper bound on this
strategy's performance is to consider only stopping at the end of
a round of such queries, even if they reveal more than $U$ strings,
which gives
\begin{align}
V \Gopt(U,V,\vW_k),
\label{eq:UB}
\end{align}
where $\Gopt(U,V,\vW_k)$ is defined in \eqref{eq:LB}.

In large deviations parlance the stochastic processes
$\{k^{-1}\log\Gopt(U,V,\vW_k)\}$ and $\{k^{-1}\log(V\Gopt(U,V,\vW_k))\}$
arising from equations \eqref{eq:LB} and \eqref{eq:UB} are exponentially
equivalent, e.g. Section 4.2.2 \cite{Dembo98}, as $\lim_{k\to\infty}
k^{-1}\log V=0$. As a result, if one process satisfies the LDP
with a rate function that has compact level sets, then the other
does \cite{Dembo98}[Theorem 4.2.3]. Thus if
$\{k^{-1}\log\Gopt(U,V,\vW_k)\}$ can be shown to satisfy a LDP,
then the round-robin strategy is proved to be asymptotically optimal.

\section{Asymptotic performance of optimal strategies}
\label{sec:results}

We first recall what is known for the single-user setting.
For each individual user $v\in\{1,\ldots,V\}$, the
specific R\'enyi entropy of the sequence $\{\Wv_k\}$, 
should it exist, is defined by
\begin{align*}
\Rv(\beta):=
	\lim_{k\to\infty} \frac 1k \frac{1}{1-\beta} 
	\log \sum_{w_k\in\A^k} P(\Wv_k=w_k)^\beta
\end{align*}
for $\beta\in(0,1)\cup(1,\infty)$, and for $\beta=1$,
\begin{align*}
\Rv(1)&:=\lim_{\beta\uparrow1} \Rv(\beta)\\
&= -\lim_{k\to\infty} \frac1k \sum_{w_k\in\A^k} P(\Wv_k=w_k)\log P(\Wv_k=w_k),
\end{align*}
the specific Shannon entropy. Should $\Rv(\beta)$ exist
for $\beta\in(0,\infty)$, then the specific min-entropy
is defined 
\begin{align*}
\Rv(\infty)&=\lim_{\beta\to\infty}\Rv(\beta) \\
&= -\lim_{k\to\infty}\frac1k \max_{w_k\in\A^k} \log P(\Wv_k=w_k),
\end{align*}
where the limit necessarily exists.
The existence of $\Rv(\beta)$ for all $\beta>0$ and its relationship
to the scaled Cumulant Generating Function
(sCGF)
\begin{align}
\LambdaGv(\alpha)
&= \lim_{k\to\infty} \frac 1k \log E(\exp(\alpha\log \Gv(\Wv_k))) \nonumber\\
	&= 
	\begin{cases}
	\displaystyle
	\alpha \Rv\left(\frac{1}{1+\alpha}\right) & \text{ if } \alpha>-1\\
	-\Rv(\infty) & \text{ if } \alpha\leq-1
	\end{cases}
\label{eq:sCGF}
\end{align}
has been established for the single user case for a broad class of
character sources that encompasses i.i.d., Markovian and general
sofic shifts that admit an entropy condition
\cite{Arikan96,Malone042,Pfister04,Hanawal11,Christiansen13}.
If, in addition, $\Rv(\beta)$ is differentiable 
with respect to $\beta$
and has a continuous
derivative, it is established in \cite{Christiansen13} that the
process $\{k^{-1}\log \Gv(\Wv_k)\}$ satisfies a LDP, 
i.e. equation \eqref{eq:LDP},
with a convex rate function
\begin{align}
\label{eq:rf}
\LambdaGv^*(x) = \sup_{\alpha\in\R}\left(x\alpha-\LambdaGv(\alpha)\right).
\end{align}
In \cite{Christiansen13}, this LDP is used to deduce an 
approximation to the guesswork distribution,
\begin{align}
\label{eq:approx1}
P(\Gv(\Wv_k)=n) \approx 
\frac 1n \exp\left(-k\LambdaGv^*\left(\frac 1k\log n\right)\right)
\end{align}
for large $k$ and $n\in\{1,\ldots,m^k\}$. 

The following theorem establishes the fundamental analogues of these
results for an asymptotically optimal strategy, where
user strings may have distinct statistical properties.

\begin{theorem}
\label{thm:main}
Assume that the components of $\{\vW_k\}$ are independent and
that for each $v\in\{1,\ldots,V\}$ $\Rv(\beta)$ exists for
all $\beta>0$, is differentiable and has a continuous derivative,
and that equation \eqref{eq:sCGF} holds. Then the process
$\{k^{-1}\log\Gopt(U,V,\vW_k)\}$, and thus any asymptotically optimal
strategy, satisfies a Large Deviation Principle.
Defining 
\begin{align*}
&\deltav(x)=\begin{cases}
	\LambdaGv^*(x) & \text{ if }x\le \Rv(1)\\
	0 & \text{ otherwise}
	\end{cases}\\
&\text{  and } \gammav(x)=\begin{cases}
	\LambdaGv^*(x) &\text{ if }x\ge \Rv(1)\\
	0 & \text{ otherwise}
\end{cases},
\end{align*}
the rate function is
\begin{align}
&\IGopt(U,V,x)=\nonumber\\
&\min_{v_1,\ldots, v_V}
\left(\LambdaGvone^*(x)+\sum_{i=2}^U \delta^{(v_i)}(x)
+\sum_{i=U+1}^V \gammavi(x)\right),
\label{eq:Iopt}
\end{align}
which is lower semi-continuous and has compact level sets, but
may not be convex. The sCGF capturing how the moments scale is 
\begin{align}
\LambdaGopt(U,V,\alpha)
&= \lim_{k\to\infty} \frac 1k \log E(\exp(\alpha\log \Gopt(U,V,\vW_k)))\nonumber\\
	&= \sup_{x\in[0,\log(m)]}\left(\alpha x -\IGopt(U,V,x)\right).  
\label{eq:scgf}
\end{align}
\end{theorem}
\begin{proof}
Under the assumptions of the theorem, 
for each $v\in\{1,\ldots,V\}$, $\{k^{-1}
\log \Gv (\Wv_k) \}$ satisfies the LDP with the rate function
given in equation \eqref{eq:rf}. As users' strings are selected
independently, the sequence of vectors
\begin{align*}
\left\{\left(\frac 1k \log G^{(1)}(W_k^{(1)}), \ldots, 
	\frac 1k \log G^{(V)}(W_k^{(V)}) \right)\right\}
\end{align*}
satisfies the LDP in $\R^V$ with rate function $I(y^{(1)},\ldots,y^{(V)})
= \sum_{v=1}^V \LambdaGv^*(y^{(v)})$, the sum of the rate functions
given in equation \eqref{eq:rf}.

Within our setting, the contraction principle, e.g. Theorem 4.2.1
\cite{Dembo98}, states that if a sequence of random variables
$\{X_n\}$ taking values in a compact subset of $\R^V$ satisfies a
LDP with rate function $I:\R^V\mapsto[0,\infty]$ and $f:\R^V\mapsto\R$
is a continuous function, then the sequence $\{f(X_n)\}$ satisfies
the LDP with rate function $\inf_{\vy}\{I(\vy):f(\vy)=x\}$.

Assume, without loss of generality, that $\vx\in\R^V$ is such that
$x^{(1)}<x^{(2)}<\cdots<x^{(V)}$, so that $\Umin(\vx) = x^{(U)}$,
and let
$\vx_n=(x^{(1)}_n,\ldots,x^{(V)}_n)\to\vx$. Let
$\epsilon < \inf\{x^{(v)}-x^{(v-1)}:v\in\{2,\ldots,V\}\}$. There
exists $N_\epsilon$ such that $\max_{v=1,\ldots,V}
|x^{(v)}_n-x^{(v)}|<\epsilon$ for all $n>N_\epsilon$. Thus
for all $v\in\{2,\ldots,V\}$ and all $n>N_\epsilon$ 
$x^{(v)}_n-x^{(v-1)}_n > x^{(v)}-x^{(v-1)}-\epsilon>0$
and so
$|\Umin(\vx_n)-\Umin(\vx)|=|x_n^{(U)}-x^{(U)}|<\epsilon$.
Hence $\Umin:\R^V\to\R$ is a continuous function
and that a LDP 
holds follows from an application of the contraction principle,
giving the rate function
\begin{align*}
\IGopt(U,V,x) = \inf\left\{
                \sum_{v=1}^V \LambdaGv^*(y_v):\Umin(y_1,\ldots,y_V)=x
                \right\}.
\end{align*}
This expression simplifies to that in equation \eqref{eq:Iopt} by
elementary arguments.
The sCGF result follows from an
application of Varahadan's Lemma, e.g \cite[Theorem 4.3.1]{Dembo98}.
\end{proof}

The expression for the rate function in equation \eqref{eq:Iopt}
lends itself to a useful interpretation. In the long string-length
asymptotic, the likelihood that an inquisitor has identified $U$
of the $V$ users' strings after approximately $\exp(kx)$ queries
is contributed to by three distinct groups of identifiable users.
For given $x$, the argument in the first term $(v_1)$ identifies
the last of the $U$ users whose string is identified. The second
summed term is contributed to by the collection of users, $(v_2)$
to $(v_U)$, whose strings have already been identified prior to
$\exp(kx)$ queries, while the final summed term corresponds to those
users, $(v_{U+1})$ to $(v_V)$, whose strings have not been identified.

The reason for using the notation $\IGopt(U,V,\cdot)$ in lieu of
$\LambdaGopt^*(U,V,\cdot)$ for the rate function in Theorem
\ref{thm:main} is that $\IGopt(U,V,\cdot)$ is not convex in general,
which we shall
demonstrate by example, and so is not always the
Legendre-Fenchel transform of the sCGF $\LambdaGopt(U,V,\cdot)$.
Instead
\begin{align*}
\LambdaGopt^*(U,V,x) = 
	\sup_\alpha\left(\alpha x -\LambdaGopt(U,V,\alpha)\right)
\end{align*}
forms the convex hull of $\IGopt(U,V,\cdot)$. In particular, this
means that we could not have proved Theorem \ref{thm:main} by
establishing properties of $\LambdaGopt(U,V,\cdot)$ alone, which
was the successful route taken for the $U=V=1$ setting, and instead
needed to rely on the LDP proved in \cite{Christiansen13}. 
Indeed, in the setting considered in \cite{Merhav99, Hanawal11a}
with $U=1$, $V=2$, with one of the strings chosen uniformly, while
the authors directly identify $\LambdaGopt(1,2,\alpha)$ for $\alpha>0$,
one cannot establish a full LDP from this approach as the resulting
rate function is not convex.  

Convexity 
of the rate function defined in equation \eqref{eq:Iopt}
is ensured, however, if all
users select 
strings
using the same stochastic properties, whereupon the
results in Theorem \ref{thm:main} simplify greatly.

\begin{corollary}
\label{cor:same}
If, in addition to the assumptions of Theorem \ref{thm:main},
$\LambdaGv(\cdot)=\LambdaG(\cdot)$ for all $v\in\{1,\ldots,V\}$ with
corresponding R\'enyi entropy $R(\cdot)$, then the rate function
in equation \eqref{eq:rf} simplifies to the convex function
\begin{align}
\LambdaGopt^*(U,V,x) &= 
	\begin{cases}
	\displaystyle
	U\LambdaG^*(x) & \text{ if } x\leq R(1)\\
	\displaystyle
	(V-U+1)\LambdaG^*(x) & \text{ if } x\geq R(1)
	\end{cases}
\label{eq:rf2}
\end{align}
where $R(1)$ is the specific Shannon entropy, and the sCGF in equation
\eqref{eq:scgf} simplifies to
\begin{align}
\label{eq:scgf2}
\LambdaGopt(U,V,\alpha) &=
	\begin{cases}
	\displaystyle
	U\LambdaG\left(\frac{\alpha}{U}\right) & \text{ if } \alpha\leq0\\
	\displaystyle
	(V-U+1)\LambdaG\left(\frac{\alpha}{V-U+1}\right) & \text{ if } \alpha\geq0.
	\end{cases}
\end{align}
In particular, with $\alpha=1$ we have
\begin{align}
&\lim_{k\to\infty} \frac 1k \log E\left(\Gopt(U,V,\vW_k)\right)\nonumber\\
	&= \LambdaGopt(1)\nonumber\\
	&= (V-U+1)\LambdaG\left(\frac{1}{V-U+1}\right)\nonumber\\
	&= R\left(\frac{V-U+1}{V-U+2}\right),
	\label{eq:average}
\end{align}
where $R((n+1)/(n+2))-R((n+2)/(n+3))$ is a decreasing
function of $n\in\N$.
\end{corollary}
\begin{proof}
The simplification in equation \eqref{eq:rf2} follows readily from
equation \eqref{eq:Iopt}. To establish that
$R((n+1)/(n+2))-R((n+2)/(n+3))$ is a decreasing function of $n\in\N$,
it suffices to establish that $R((x+1)/(x+2))$ is a convex, decreasing
function for $x\in\R_+$.

That $R(x)\downarrow R(1)$ as $x\uparrow1$ is a general property
of specific R\'enyi entropy. For convexity, using equation
\eqref{eq:average} it suffices to
show that $x\LambdaG(1/x)$ is convex for $x>0$. This can be seen
by noting that for any $a\in(0,1)$ and $x_1,x_2>0$,
\begin{align*}
&(a x_1+(1-a)x_2)
\LambdaG\left(\frac{1}{a x_1+(1-a)x_2}\right)\\
&= (a x_1+(1-a)x_2) 
	\LambdaG\left(\eta \frac{1}{x_1}+(1-\eta)\frac{1}{x_2}\right)\\
&\leq a x_1 \LambdaG\left(\frac{1}{x_1}\right)
	+(1-a)x_2 \LambdaG\left(\frac{1}{x_2}\right),
\end{align*}
where $\eta = a x_1/(a x_1+(1-a)x_2)\in(0,1)$ and we have used the
convexity of $\LambdaG$. 
\end{proof}

As the growth rate, $R((n+1)/(n+2))-R((n+2)/(n+3))$, is decreasing
there is a law of diminishing returns for the inquisitor where the
greatest decrease in the average guesswork growth rate is through
the provision of one additional user. 
From the system designer's point of view, the specific Shannon
entropy of the source is a universal lower bound on the exponential
growth rate of the expected guesswork that, while we cannot take
the limit to infinity, is tight for large $V-U$.

Regardless 
of whether the rate function $\IGopt(U,V,\cdot)$ is
convex, Theorem \ref{lem:approx}, which follows, justifies the approximation
\begin{align*}
P(\Gopt(U,V,\vW_k)=n) \approx 
	\frac 1n \exp\left(-k\IGopt\left(U,V,\frac 1k\log n\right)\right)
\end{align*}
for large $k$ and $n\in\{1,\ldots,m^k\}$. It is analogous to that
in equation \eqref{eq:approx1}, first developed in \cite{Christiansen13},
but there are additional difficulties that must be overcome to
establish it. In particular, if $U=V=1$, the likelihood that the
string is identified at each query is a decreasing function of guess
number, but this is not true in the more general case.

As a simple example, consider $U=V=2$, $\A=\{0,1\}$, strings of
length $1$ and strings chosen uniformly. Here the probability of
guessing both strings in one guess is $1/4$, but at the second guess
it is $3/4$. Despite this lack of monotonicity, the approximation
still holds in the following sense.

\begin{theorem}
\label{lem:approx}
Under the assumptions of Theorem \ref{thm:main}, for any $x \in [0,
\log m)$ we have
\begin{align*}
&\lim_{\epsilon \downarrow 0}\liminf_{k\rightarrow \infty}
	\frac 1k \log \inf_{n \in K_k(x, \epsilon)}P(\Gopt(U,V,\vW_k)=n)\\
&=\lim_{\epsilon \downarrow 0}\limsup_{k\rightarrow \infty}
	\frac 1k \log \sup_{n \in K_k(x, \epsilon)}P(\Gopt(U,V,\vW_k)=n)\\
&=-I_{\Gopt}(U, V, x)-x,
\end{align*}
where
\begin{align*}
K_k(x, \epsilon)=\{n:n \in (\exp(k(x-\epsilon)), \exp((k(x+\epsilon)))\}
\end{align*}
is the collection of guesses made in a log-neighborhood of $x$.
\end{theorem}
\begin{proof}
The proof follows the ideas in \cite{Christiansen13} Corollary 4,
but with the added difficulties resolved by isolating the last word
that is likely to be guessed and leveraging the monotonicity it its
individual likelihood of being identified.

Noting the definition of $K_k(x, \epsilon)$ in the statement
of the theorem, consider for $x\in(0,\log(m))$
\begin{align*}
&\sup_{n \in K_k(x, \epsilon)}P(\Gopt(U,V,\vW_k)=n)\nonumber\\
=& \sup_{n \in K_k(x, \epsilon)}\sum_{(v_1, \ldots, v_V)}
	P(\Gvone(W_k^{(v_1)})=n)\\
	&\prod_{i=2}^{U}P(\Gvi(W_k^{(v_i)})\le n)
	\prod_{i=U+1}^{V}P(\Gvi(W_k^{(v_i)})\ge n)\\
\le& \sup_{n \in K_k(x, \epsilon)}
	\max_{(v_1, \ldots, v_V)}(V!)P(\Gvone(W_k^{(v_1)})=n)\\
	&\prod_{i=2}^{U}P(\Gvi(W_k^{(v_i)})\le n)
		\prod_{i=U+1}^{V}P(\Gvi(W_k^{(v_i)})\ge n)\\
\le& \sup_{n \in K_k(x, \epsilon)}
	\max_{(v_1, \ldots, v_V)}(V!)P(\Gvone(W_k^{(v_1)})=n)\\
	&\prod_{i=2}^{U}
	P\left(\frac 1k \log \Gvi(W_k^{(v_i)})\le x-\epsilon\right)\\
&\prod_{i=
U+1
}^{V}P\left(\frac 1k \log \Gvi(W_k^{(v_i)})\ge x+\epsilon\right)\nonumber\\
&\le \inf_{n \in K_k(x-2\epsilon, \epsilon)}\max_{(v_1, \ldots, v_V)}(V!)
P\left(\frac1k \log \Gvone(W_k^{(v_1)})=n\right)\\
&\prod_{i=2}^{U}
P\left(\frac 1k \log \Gvi(W_k^{(v_i)})\le x+\epsilon\right)\\
&\prod_{i=U+1}^{V}P\left(\frac 1k \log\Gvi(W_k^{(v_i)})\ge x-\epsilon\right).\nonumber
\end{align*}
The first equality holds by definition of $\Gopt(U,V,\cdot)$.  The
first inequality follows from the union bound over all possible
permutations of $\{1, \ldots, V\}$. The second inequality utilizes
$k^{-1}\log n \in (x-\epsilon, x+\epsilon)$ if $n\in K_k(x, \epsilon)$,
while the third inequality uses the monotonic decreasing probabilities
in guessing a single user's string.

Taking $\lim_{\epsilon \downarrow 0}\limsup_{k\rightarrow
\infty}k^{-1}\log$ on both sides of the inequality, interchanging
the order of the max and the supremum, using the continuity
of $\LambdaGv(\cdot)$ for each $v\in\{1,\cdots,V\}$, and the
representation of the rate function $\IGopt(U,V,\cdot)$ in equation
\eqref{eq:Iopt}, gives the upper bound
\begin{align*}
&\lim_{\epsilon \downarrow 0}\limsup_{k\rightarrow \infty}\frac 1k \log \sup_{n \in K_k(x, \epsilon)}
	P(\Gopt(\vW_k)=n) \\
&\leq-\IGopt(U, V, x)-x.
\end{align*}
Considering the least likely guesswork in the ball leads to a
matching lower bound. The other case, $x=0$, follows similar logic,
leading to the result.
\end{proof}

We next provide some illustrative examples of what these results
imply, returning to using $\log_2$ in figures.

\section{Mismatched Statistics Example}
\label{sec:mismatch}

The potential lack of convexity in the rate function of Theorem
\ref{thm:main}, equation \eqref{eq:Iopt}, only arises if users'
string statistics are asymptotically distinct. The significance of
this lack of convexity on the phenomenology of guesswork can be
understood in terms of the asymptotically optimal round-robin
strategy: if the rate function is not convex, there is no single
set of users whose strings are most vulnerable.
That is, if $U$ strings are recovered after a small number of
guesses, they will be from one set of users, but after a number of
guesses corresponding to a transition from the initial convexity
they will be from another set of users. 
This is made explicit in the following corollary to Theorem
\ref{thm:main}.

\begin{corollary}
If $\IGopt(U,V,x)$ is not convex in $x$, then there is there is no
single set of users whose strings will be identified in the long
string length asymptotic.
\end{corollary}
\begin{proof}
We prove the result by establishing the contraposition: if a single set
of users is always most vulnerable, then $\IGopt(U,V,x)$ is convex.
Recall the expression for $\IGopt(U,V,x)$ given in equation
\eqref{eq:Iopt}
\begin{align*}
&\IGopt(U,V,x)=\\
&\min_{v_1,\ldots, v_V}
\left(\LambdaGvone^*(x)+\sum_{i=2}^U \delta^{(v_i)}(x)
+\sum_{i=U+1}^V \gammavi(x)\right),
\end{align*}
As explained after Theorem \ref{thm:main}, for given $x$ the set
of users $\{(v_1),\ldots,(v_U)\}$ corresponds to those users whose
strings, on the scale of large deviations, will be identified by
the inquisitor after approximately $\exp(kx)$ queries. If this set
is unchanging in 
$x$, i.e. the same set of users is identified
irrespective of $x$, then both of the functions
\begin{align*}
\left(\LambdaGvone^*(x)+\sum_{i=2}^U \delta^{(v_i)}(x)\right)
\text{ and }
\sum_{i=U+1}^V \gammavi(x)
\end{align*}
are sums of functions that are convex in $x$, and so are convex
themselves. Thus the sum of them, $\IGopt(U,V,x)$, is convex.
\end{proof}

This is most readily
illustrated by an example that falls within the two-user setting
of \cite{Merhav99}, where one string is constructed from uniformly
from i.i.d. bits and the other string from non-uniformly selected
i.i.d. bytes.
\begin{figure}
\begin{center}
\includegraphics[scale=0.46]{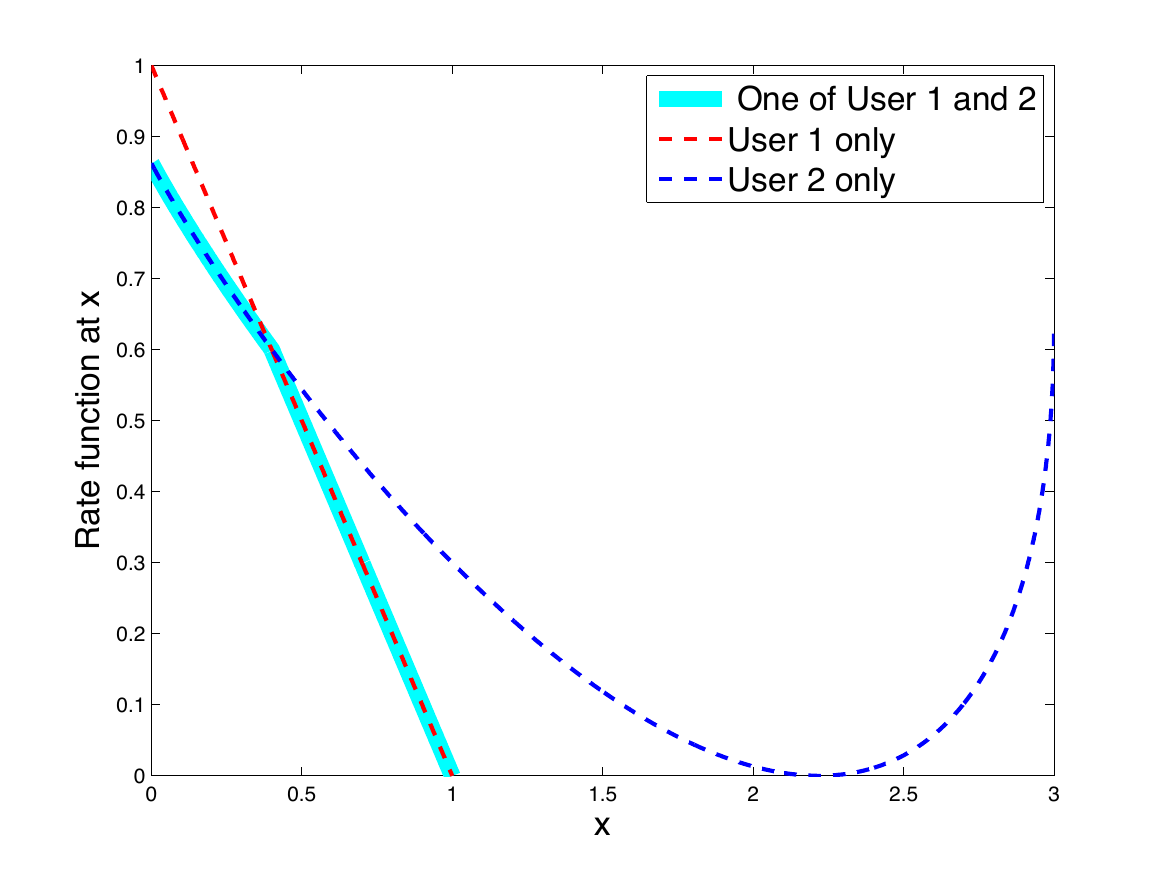}
\end{center}
\caption{User 1 picks a uniform bit string. User 2 picks a non-uniform
i.i.d. byte string. The straight line starting at $(0,1)$ displays
${\LambdaG^{(1)}}^*(x)$, the large deviations rate function for
guessing the uniform bit string. The convex function starting
below it is ${\LambdaG^{(2)}}^*(x)$, the rate function for guessing
the non-uniform byte string. The highlighted line, which is the
minimum of the two rate functions until $x=1$ and then $+\infty$
afterwards, displays $I_{\Gopt}(1, 2, x)$, as determined by
\eqref{eq:Iopt}, the rate function for an inquisitor to guess one
of the two strings. Its non-convexity demonstrates that initially
it is the bytes that are most likely to be revealed by brute force
searching, but eventually it is the uniform bits that are more
likely to be identified. The Legendre-Fenchel transform of the
scaled cumulant generating function of the guesswork distributions
would form the convex hull of the highlighted line and so this
could not be deduced by analysis of the asymptotic moments.
}
\label{fig:nonconvex}
\end{figure}

Let $\A=\{0,\ldots,7\}$, $U=1$ and $V=2$. Let one character source
correspond to the output of a cryptographically secure pseudo-random
number generator. That is, despite having a byte alphabet, the source
produces perfectly uniform i.i.d. bits,
\begin{align*}
P(W^{(1)}_1=i) &=
	\begin{cases}
	1/2 & \text{ if } i\in\{0,1\}\\
	0 & \text{ otherwise}.
	\end{cases}
\end{align*}
The other source can be thought of as i.i.d. bytes generated by a
non-uniform source,
\begin{align*}
P(W^{(2)}_1=i) &=
	\begin{cases}
	0.55 & \text{ if } i=0\\
	0.1 & \text{ if } i\in\{1,2\}\\
	0.05 & \text{ if } i\in\{3,\ldots,7\}.
	\end{cases}
\end{align*}
This models the situation of a piece of data, a string from the
second source, being encrypted with a shorter, perfectly uniform
key. The inquisitor can reveal the hidden string by guessing either
the key or the string. One might suspect that either the key or
the string is necessarily more susceptible to being guessed, but
the result is more subtle.

Figure \ref{fig:nonconvex} plots the rate functions for guessing
each of the user's strings individually as well as the rate function
for guessing one out of two, determined by equation \eqref{eq:Iopt},
which in this case is the minimum of the two rate function where
they are finite. The y-axis is the exponential decay-rate in string
length $k$ of the likelihood of identification given approximately
$\exp(k x)$ guesses, where $x$ is on the x-axis, have been made.
The rate function reveals that if the inquisitor identifies one of
the strings quickly, it will be the non-uniform byte string, but
after a certain number of guesses it is the key, the uniform bit
string, that is identified.

Attempting to obtain this result by taking the Legendre Fenchel
transform of the sCGF identified in \cite{Merhav99} results in the
convex hull of this non-convex function, which has no real meaning.
This explains the necessity for the distinct proof approach taken
here if one wishes to develop estimates on the guesswork distribution
rather than its moments.

\section{Identical Statistics Examples}
\label{sec:ident}

When the string statistics of users are asymptotically the same,
the resulting multi-user guesswork rate functions are convex by
Corollary \ref{cor:same}, and the r\^ole of specific Shannon entropy
in analyzing expected multi-user guesswork appears.  This is the
setting that leads to the results in Section \ref{sec:firstexample}
where it is assumed that character statistics are i.i.d., but not
necessarily uniform.

An alternate means of departure from string-selection uniformity
is that the appearance of characters within the string may be
correlated. The simplest model of this is where string symbols are
governed by a Markov chain with arbitrary starting distribution and
transition matrix
\begin{align*}
\left(
	\begin{array}{cc}
	1-a & a \\
	b & 1-b 
	\end{array}
\right),
\end{align*}
where $a,b\in(0,1)$. The specific R\'enyi entropy of this character
source can be evaluated, e.g. \cite{Malone042}, for $\beta\neq1$ to be
\begin{align*}
R(\beta) =& \frac{1}{1-\beta}\log
	\left(
	(1-a)^\beta+(1-b)^\beta \right. \\
	& \left.+
	\sqrt{((1-a)^\beta-(1-b)^\beta)^2 + 4(ab)^\beta } 
	\right)
	-\frac{1}{1-\beta}
\end{align*}
and $R(1)$ is the Shannon entropy
\begin{align*}
R(1) &=  
\frac{b}{a+b}H(a)+\frac{a}{a+b}H(b),
\end{align*}
where $H(a) = -a\log(a)-(1-a)\log(1-a)$.

Figure \ref{fig:gap_mark} shows $R(1/2)-R(1)$ the difference between
the average guesswork growth rate for a single user system versus
one for an arbitrarily large number of users as $a$ and $b$ are
varied. Heavily correlated sources or those with unlikely characters
give the greatest discrepancy in security. 

\begin{figure}
\begin{center}
\includegraphics[scale=0.46]{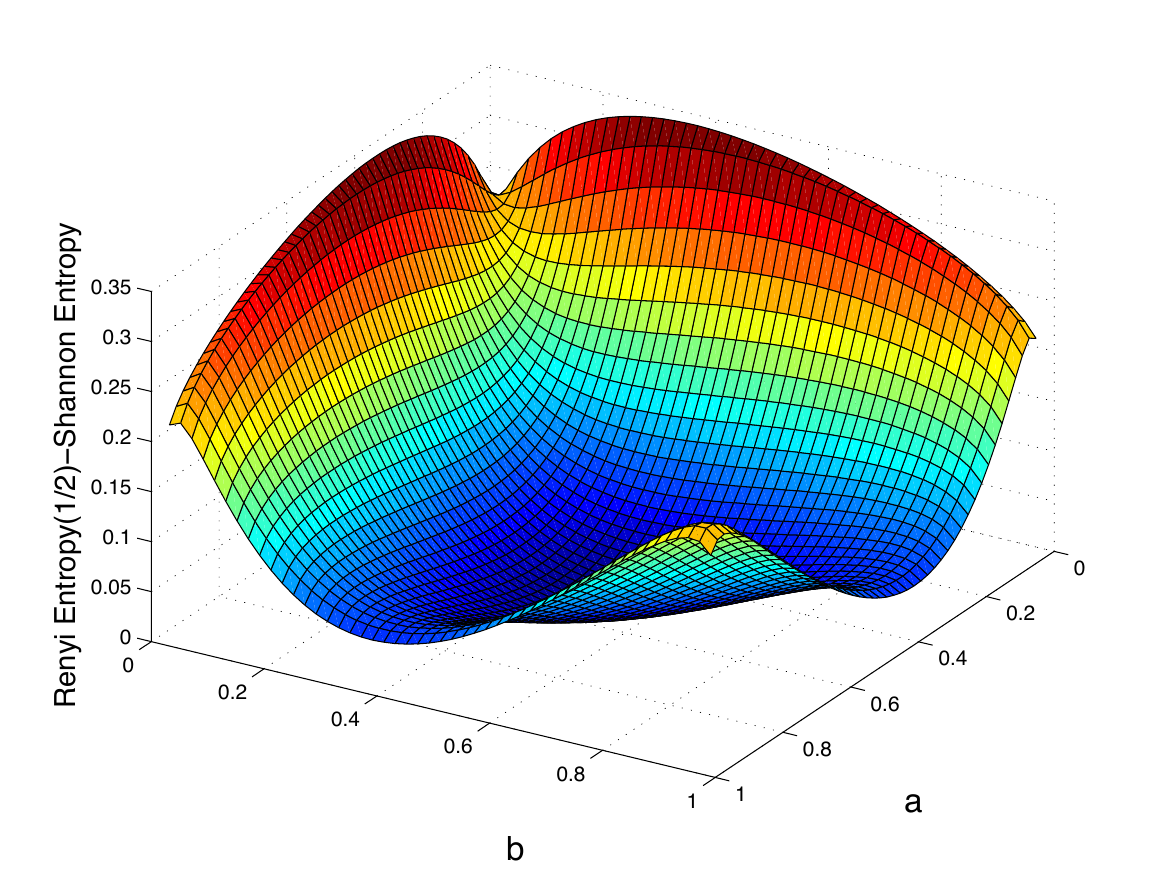}
\end{center}
\caption{Markovian string source over a binary alphabet $\A=\{0,1\}$
with $a$ being the probability of a 1 after a 0 and $b$ being the
probability of a 0 after a 1. The plots shows the difference in
average guesswork exponent for a single user system and a system
with an arbitrarily large number of users, a measure of computational
security reduction.
}
\label{fig:gap_mark}
\end{figure}
If $a=b$, then the stationary likelihood a symbol is a $0$ or $1$
is equal, but symbol occurrence is correlated. In that setting, the
string source's specific R\'enyi entropy gives for $\beta\neq1$
\begin{align*}
R(\beta) =& \frac{1}{1-\beta}\log
	\left(
	(1-a)^\beta + a^\beta
	\right),
\end{align*}
which is the same as a Bernoulli source with probability $a$ of one
character. Thus the results in Section \ref{sec:firstexample} can
be re-read with the Bernoulli string source with parameter $p=a$
substituted for a Markovian string source whose stationary distribution
gives equal weight to both alphabet letters, but for which character
appearance is correlated.

\section{Discussion}
\label{sec:conc}

Since Massey \cite{Massey94} posed the original guesswork problem
and Arikan \cite{Arikan96} introduced its long string asymptotic,
generalizations have been used to quantify the computational security
of several systems, including being related to questions of loss-less
compression. Here we have considered what appears to be one of the
most natural extensions of that theory, that of multi-user computational
security. As a consequence of the inherent non-convex nature of
the guesswork rate function unless string source statistics are
equal for all users, this development wasn't possible prior to the
Large Deviation Principle proved in \cite{Christiansen13}. The
results therein themselves relied on the earlier work that determined
the scaled cumulant generating function for the guesswork for a
broad class of process \cite{Arikan96,Malone042,Pfister04,Hanawal11}.

The fact that rate functions can be non-convex encapsulates that
distinct subsets of users are likely to be identified depending on
how many unsuccessful guesses have been made. As a result, a simple
ordering of string guessing difficulty is inappropriate in multi-user
systems and suggests that quantification of multi-user computational
security is inevitably nuanced.

The original analysis of the asymptotic behavior of single user
guesswork identified an operational meaning to specific R\'enyi
entropy. In particular, the average guesswork grows exponentially
in string length with an exponent that is the specific R\'enyi
entropy of the character source with parameter $1/2$. When users'
string statistics are the same, the generalization to multi-user
guesswork identifies a surprising operational r\^ole for specific
R\'enyi entropy with parameter $n/(n+1)$ for each $n\in\N$ when
$n$ is the excess number of strings that can be guessed. Moreover,
while the specific Shannon entropy of the string source was found
in the single user problem to have an unnatural meaning as the
growth rate of the expected logarithm of the guesswork, in the
multi-user system it arises as the universal lower bound on the
average guesswork growth rate.

For the asymptote at hand, the key message is that there is a law
of diminishing returns for an inquisitor as the number of users
increases. For a multi-user system
designer,
in contrast to the single character, single user system 
introduced in \cite{Massey94},
Shannon entropy is the appropriate measure of 
expected guesswork for systems with many users.

Future work might consider the case where the $V$ strings are not
selected independently, as was assumed here, but are instead linear
functions of $U$ independent strings. A potential application of
such a case, suggested by Erdal Arikan (Bilkent University) in a
personal communication, envisages the use of multi-user guesswork
to characterize the behavior of parallel concatenated decoders
operating on blocks of convolutionally encoded symbols passed though
a preliminary algebraic block Maximum Distance Separable (MDS) code,
e.g. \cite{shu2004}. The connection between guessing and convolutional
codes was first established by Arikan \cite{Arikan96}.

Decoding over a channel may, in general, be viewed as guessing a
codeword that has been chosen from a list of possible channel input
sequences, given the observation of an output sequence formed by
corrupting the input sequence according to some probability law
used to characterize the channel, e.g \cite{Christiansen13b}.
Considering sequential decoding of convolutional codes, first
proposed by Wozencraft \cite{Wozencraft57}, that guessing may
constitute an exploration along a decision tree of the possible
input sequences that could have led to the observed output sequence,
as modeled by Fano \cite{Fano63}.  If the transmitted rate, given
by the logarithm of the cardinality of possible codewords, falls
below the cut-off rate, then results in \cite{Arikan96} prove that
the guesswork remains in expectation less than exponential in the
length of the code. Beyond the cut-off rate, it becomes exponentially
large. One may view such a result as justifying the frequent use
of cut-off rate as as a practical, engineering characterization of
the limitations of block and convolutional codes.

Consider now the following construction of a type of concatenated
code \cite{shu2004}, which is a slight variant of that proposed by
Falconer \cite{Falconer67}. The original data, a stream of i.i.d.
symbols, is first encoded using an algebraic block MDS code. For a
block MDS code, such as a Reed-Solomon code \cite{shu2004}, over a
codeword constituted by a sequence of $V$ symbols, correct reception
at the output of any $U$ symbols from the $V$ allows for correct
decoding, where the feasibility of a pair of $V$ and $U$ depends
on the family of codes. For every $U$ input symbols in the data
stream, $V$ symbols are generated by the algebraic block MDS code.
Note that these symbols may be selected over a set of large
cardinality, for instance by taking each symbol to be a string of
bits. As successive input blocks of length $U$ are processed by the
block MDS code, these symbols form $V$ separate streams of symbols.
Each of these $V$ streams emanating from the algebraic block MDS
code is coded using a separate but identical convolutional encoder.

The $V$ convolutional codewords thus obtained are dependent, even
though any $U$ of them are mutually independent. This dependence
is imputed by the fact that the $V$ convolutional codewords are
created by $U$ original streams that form the input of the block
MDS encoder. The $V$ convolutional codewords constitute then the
inputs to $V$ mutually independent, Discrete Memoryless Channels
(DMCs), all governed by the same probability law. In Falconer's
construct, such parallel DMCs are embodied by time-sharing equally
a single DMC. While Falconer envisages independent DMCs governed
by a single probability law, as is suitable in the setting of
interleaving over a single DMC, we may readily extend the scheme
to the case where the parallel DMCs have different behaviors. Such
a model is natural in wireless settings where several channels are
used in parallel, say over different frequencies. While the behavior
of such channels is often well modeled as being mutually independent,
and the channels individually are well approximated as being DMCs,
the characteristics of the channels, which may vary slowly in time,
generally differ considerably from each other at any time.

Decoding
uses the outputs of the $V$ DMCs as follows. For each DMC,
the output is initially individually decoded using sequential
decoding so that, in the words of Falconer, "controlled by the Fano
algorithm, all $[V]$ sequential decoders simultaneously and
independently attempt to proceed along the correct path in their
own trees". The dependence among the streams produced by the original
application of the block MDS code entails that, when $U$ sequential
decoders each correctly guesses a symbol, the correct guesses
determine a block of $U$ original data symbols. The latter are
communicated to all remaining $V-U$ sequential decoders, eliminating
the need for them to continue producing guesses regarding that block
of $U$ original data symbols. The sequential decoders then proceed
to continue attempting to decode the next block of $U$ original
data symbols. This scheme allows the $U$ most fortunate guesses out
of $V$ to dominate the performance of the overall decoder. A
sequential decoder that was a laggard for one block of the original
$U$ symbols may prove to be a leader for another block of $U$
symbols.

\section*{Acknowledgments:}
The authors thank Erdal Arikan (Bilkent University) for informative
feedback and for pointing out the relationship between multi-user
guesswork and sequential decoding. 
They also thank the anonymous reviewers for their feedback 
on the paper.


\end{document}